\documentclass[copyright,creativecommons]{eptcs}
\providecommand{\event}{DICE-FOPARA 2019} 
\usepackage{underscore}           
\usepackage[T5,T1]{fontenc}
\usepackage[utf8]{inputenc}

\title{On the Elementary Affine Lambda-calculus\\
  with and Without Type Fixpoints}
\author{{\fontencoding{T5}\selectfont Lê Thành Dũng \textsc{Nguyễn}}\thanks{Partially
  supported by the ANR project ELICA
  (ANR-14-CE25-0005).}
\institute{LIPN, UMR 7030 CNRS, Université Paris 13, France}
\email{nltd@nguyentito.eu}
}
\def\titlerunning{On the elementary affine $\lambda$-calculus with and without
  type fixpoints}
\def\authorrunning{{\fontencoding{T5}\selectfont Nguyễn} L.\ T.\ D.}

\usepackage{amsmath}
\usepackage{amssymb}
\usepackage{amsthm}

\usepackage{cmll}
\usepackage{stmaryrd}
\usepackage{csquotes}

\theoremstyle{plain}
\newtheorem{theorem}{Theorem}[section]
\newtheorem{proposition}[theorem]{Proposition}

\newtheorem{lemma}[theorem]{Lemma}

\theoremstyle{definition}
\newtheorem{definition}[theorem]{Definition}
\newtheorem{remark}[theorem]{Remark}

\newcommand{\Str}{\mathtt{Str}}
\newcommand{\Bool}{\mathtt{Bool}}
\newcommand{\Nat}{\mathtt{Nat}}

\newcommand{\STlam}{{\mathrm{ST}\lambda}}
\newcommand{\EAlam}{{\mathrm{EA}\lambda}}
\newcommand{\muEAlam}{{\mu\mathrm{EA}\lambda}}
\newcommand{\Alamtwo}{{\mathrm{A}\lambda{}2}}

\newcommand{\Lang}{\mathcal{L}}
\newcommand{\End}{\mathrm{End}}
\newcommand{\denot}[1]{\left\llbracket #1 \right\rrbracket}
\newcommand{\trunc}[1]{\|#1\|_0}
\newcommand{\bloup}[2]{\Phi_{#1}(#2)}

\begin{document}
\maketitle

\begin{abstract}
  The elementary affine $\lambda$-calculus was introduced as a polyvalent
  setting for implicit computational complexity, allowing for characterizations
  of polynomial time and hyperexponential time predicates. But these results
  rely on type fixpoints (a.k.a.\ recursive types), and it was unknown whether
  this feature of the type system was really necessary. We give a positive
  answer by showing that without type fixpoints, we get a characterization of
  regular languages instead of polynomial time. The proof uses the semantic
  evaluation method. We also propose an aesthetic improvement on the
  characterization of the function classes FP and $k$-FEXPTIME in the presence
  of recursive types.
\end{abstract}

\section{Introduction}

\paragraph{The elementary affine $\lambda$-calculus}

Elementary Linear Logic (ELL), introduced by Girard~\cite{GirardELL}, is a logic
that can be seen as a typed functional programming language through the
proof-as-programs correspondence. Its typing rules ensure that a function can be
expressed if and only if it is elementary recursive (as is expounded in detail
in~\cite{DanosJoinet}), hence the name. (This is an instance of the
\enquote{type-theoretic} or \enquote{Curry--Howard} approach to implicit
computational complexity.) This was refined by Baillot~\cite{baillot} into a
characterization of each level of the $k$-EXPTIME hierarchy, in an affine
variant of ELL.

A later improvement by Baillot, De Benedetti and Ronchi~\cite{Benedetti}
consisted in turning this logic into an actual type system for a functional
calculus with good properties (e.g.\ subject reduction), called the
\emph{elementary affine $\lambda$-calculus}. In this paper, we shall call their
system $\muEAlam$ -- the reason for the $\mu$ will soon become clear. The main
result about it is:
\begin{theorem}[{\cite{Benedetti}}]
  \label{thm:benedetti}
  The programs of type $\oc\Str \multimap \oc^{k+2}\Bool$ in $\muEAlam$ decide
  exactly the languages in the class \emph{$k$-EXPTIME}. In particular $\oc\Str
  \multimap\ \oc\oc\Bool$ corresponds to \emph{polynomial time (P)} predicates.
\end{theorem}
Here are some indications for the reader unfamiliar with linear or
affine type systems:
\begin{itemize}
\item a program of type $A \multimap B$ uses its input of type $A$ at
  most once to produce its output of type $B$;
\item $\oc A$ means roughly \enquote{as many $A$'s as you want}, so a function
  which uses its argument multiple times can be given a type of the form $\oc A
  \multimap B$;
\item in usual linear or affine logic, one can convert a $\oc A$ into a $A$;
  however, in the elementary affine $\lambda$-calculus, there is a restriction
  which makes the \emph{exponential depth} (number of `$\oc$' modalities)
  meaningful, one cannot perform such a depth-changing operation -- this is why
  the depth $k$ of the output $\oc^{k} \Bool$ (i.e.\ $\oc( \ldots (\oc \Bool))$ with
  $k$ `$\oc$') controls the complexity;
\item the type of booleans is defined as $\Bool = \forall \alpha.\, \alpha
  \multimap \alpha \multimap \alpha$, and it has two inhabitants;
\item $\Str = \forall \alpha.\, \Str[\alpha]$, with $\Str[\alpha] = \oc(\alpha
  \multimap \alpha) \multimap \oc(\alpha \multimap \alpha) \multimap \oc(\alpha
  \multimap \alpha)$, is the type of \emph{Church encodings of binary strings}:
  the string $w_1 \ldots w_n \in \{0,1\}^*$ is represented as the function
  which, for any type $A$, takes as input $f_0 : A \multimap A$ and $f_1 : A
  \multimap A$, and returns $f_{w_1} \circ \ldots \circ f_{w_n}$.
\end{itemize}

\paragraph{Type fixpoints and Scott encodings}

We wish to draw attention to a particular feature of this language: the presence
of \emph{type fixpoints}\footnote{A remark for the readers acquainted with typed
  $\lambda$-calculi: there is no \enquote{positivity} constraint imposed, yet
  those recursive types are harmless for the normalization property, as the
  untyped version of the elementary affine $\lambda$-calculus is already
  normalizing. The analogous property for ELL was already remarked
  in~\cite{GirardELL}.}, a.k.a.\ \emph{recursive types}. An example is the type
of \emph{Scott binary strings}:
\[ \Str_S := \forall \alpha.\, (\Str_S \multimap \alpha) \multimap (\Str_S
  \multimap \alpha) \multimap \alpha \multimap \alpha \]
In the elementary affine $\lambda$-calculus as defined in~\cite{Benedetti}, this
recursive equation can be turned into a valid type definition, by using a fixed
point operator $\mu$ on types (this explains our name $\muEAlam$):
\[ \Str_S := \mu\beta.\, \forall \alpha.\, (\beta \multimap \alpha) \multimap
  (\beta \multimap \alpha) \multimap \alpha \multimap \alpha \]
The idea is that strings are represented by their \enquote{pattern-matching}
function (destructor): if $u$ is a Scott binary string, then $u\;f_0\;f_1\;x$
morally means \enquote{if $u$ represents the empty word, return $x$; else,
  return $f_c$ applied to $v$ where $c \in \{0,1\}$ is the first letter and $v$
  represents the suffix}. Formally, we associate to each string $w \in
\{0,1\}^*$ a $\muEAlam$-term $S(w)$ of type $\Str_S$:
\[ S(\varepsilon) = \lambda f_0.\,\lambda f_1.\,\lambda x.\, x \quad
  S(0 \cdot w') = \lambda f_0.\,\lambda f_1.\,\lambda x.\, f_0\;S(w') \quad
  S(1 \cdot w') = \lambda f_0.\,\lambda f_1.\,\lambda x.\, f_1\;S(w') \]
This encoding of strings has been used to give a characterization of function
classes in $\muEAlam$:
\begin{theorem}[{\cite{Benedetti}}]
  \label{thm:scottfp}
  The programs of type $\oc\Str \multimap \oc^{k+2}\Str_S$ in $\muEAlam$ compute
  exactly the functions in the class \emph{$k$-FEXPTIME}. In particular $\oc\Str
  \multimap\ \oc\oc\Str_S$ corresponds to \emph{FP}.
\end{theorem}

\paragraph{Our contributions}

There are two natural questions concerning the necessity of type fixpoints:
\begin{itemize}
\item In the interface: it is possible to characterize this hierarchy of
  function classes using a function type involving only Church encodings?
\item In the implementation: the extensional completeness proof for the
  predicate classes (Theorem~\ref{thm:benedetti}) makes use of the type $\Str_S$
  (to represent configurations of Turing machines), even though this type does
  not appear in the statement; could one avoid recursive types in the proof?
  This question has been raised by Baillot in the conclusion of~\cite{baillot}.
\end{itemize}

In this paper, we answer both questions. The first one has a positive answer:
\begin{theorem}
  \label{thm:aesthetic}
  The programs of type $\oc\Str \multimap \oc^{k+1}\Str$ in $\muEAlam$ compute
  exactly the functions in the class \emph{$k$-FEXPTIME}. In particular $\oc\Str
  \multimap\ \oc\Str$ corresponds to \emph{FP}.
\end{theorem}

An advantage of this characterization is that it reflects the fact that
composing a $k$-FEXPTIME function $f$ with a $l$-FEXPTIME function $g$ gives a
$(k+l)$-FEXPTIME function: since any $\muEAlam$-term of type $A \multimap B$
lifts to a term of type $\oc^{k}A \multimap \oc^{k}B$ (this is called
\enquote{functorial promotion}, cf.\ Proposition~\ref{prop:promotion}), we can
compose the terms $f : \oc\Str \multimap \oc^{k+1}\Str$ and $g^{(k)} :
\oc^{k+1}\Str \multimap \oc^{(l+1)+k}\Str$ to obtain a term of type $\oc\Str
\multimap \oc^{(k+l)+1}\Str$. In particular FP is closed under composition. A
characterization of FP in $\muEAlam$ by a function type whose input and output
types coincide was proposed in~\cite{Benedetti}, but it is less natural: a
string is represented as a pair of its length (Church-encoded) and its contents
(Scott-encoded).

As for the second question, we should first mention that Girard's original
characterization of elementary recursive functions in ELL does not involve type
fixpoints. This can be replayed in the elementary affine $\lambda$-calculus
\emph{without} type fixpoints, which we shall denote by $\EAlam$.
\begin{theorem}[{\cite{baillot}}]
  The class of elementary recursive functions is the union, over $k \in
  \mathbb{N}$, of the classes of functions computed by programs of type $\oc\Str
  \multimap \oc^{k}\Str$ in $\EAlam$.
\end{theorem}
(The detailed proof given in~\cite{baillot} is for Elementary Affine Logic; it
can be directly transposed to $\EAlam$.)

However, the characterization of P by $\oc\Str \multimap \oc\oc\Bool$ fails in
$\EAlam$, as we show:
\begin{theorem}
  \label{thm:elam-reg}
  The programs of type $\oc\Str \multimap \oc\oc\Bool$ in $\EAlam$ decide exactly
  the \emph{regular languages}. This is also the case for the $\EAlam$-terms of
  type $\Str \multimap \oc\Bool$.
\end{theorem}

This result is surprising for a few reasons: the class of languages obtained is
unexpectedly small, and it hints at connections between $\EAlam$ and formal
language theory (the conclusion will discuss this further). The proof techniques
for the above theorem are quite different from those used in~\cite{Benedetti}:
instead of bounding the syntactic normalization process, we take inspiration
from the tradition of implicit complexity in the \emph{simply typed
  $\lambda$-calculus} ($\STlam$), in particular from:
\begin{theorem}[Hillebrand \& Kanellakis~\cite{HillebrandKanellakis}]
  \label{thm:hk}
  In the \emph{simply typed $\lambda$-calculus}, the languages decided by terms
  of type $\Str_{\STlam}[A] \to \Bool_\STlam$ -- $A$ is a simple type that may be chosen
  depending on the language -- are exactly the regular languages.
\end{theorem}
Here $\Str_\STlam[A] = (A \to A) \to (A \to A) \to (A \to A)$ and
$\Bool_\STlam = o \to o \to o$, where $o$ is a base type. This is proved using
the \emph{semantic evaluation} method (see~\cite{TeruiSemantics} and references
therein). To make this method work in our case, we need a new result in
denotational semantics:
\begin{lemma}
  \label{lem:finsem}
  The \emph{second-order affine $\lambda$-calculus} $\Alamtwo$ -- i.e.\ the
  subsystem of $\EAlam$ without the exponential modality `$\oc$' -- admits a
  non-trivial \emph{finite} semantics.
\end{lemma}

By \enquote{non-trivial} we mean distinguishing the two inhabitants of $\Bool =
\forall \alpha.\, \alpha \multimap \alpha \multimap \alpha$. The term
\enquote{second-order} refers to the (impredicative) polymorphism supported by
both $\muEAlam$ and $\EAlam$ -- indeed, the types $\Bool$, $\Str$ and $\Str_S$
all contain second-order quantifiers ($\forall$). The lemma means morally that
one cannot represent infinite data types in $\muEAlam$ without using the
exponential modality -- whereas in $\muEAlam$, the exponential-free type
$\Str_S$ encodes the infinite set $\{0,1\}^*$.

Thus, motivated by this question in implicit complexity, we set out to establish
the above lemma, and came up with two approaches:
\begin{itemize}
\item a \enquote{category-theoretic} solution consists in showing the finiteness
  of a pre-existing model based on coherence spaces and normal functors; this is
  the subject of another paper~\cite{FiniteMALL2};
\item a \enquote{syntactic} solution, developed in a joint work with P.\
  Pistone, T.\ Seiller and L.\ Tortora de Falco, relies on a careful
  combinatorial study of second-order proof nets; it will be written up in an
  upcoming paper.
\end{itemize}
The further development of these semantic tools has led to more results on
$\EAlam$ and/or on Elementary Linear Logic without type fixpoints, which are
beyond the scope of the present paper. This includes an already published joint
work with P.\ Pradic~\cite{sequel} on logarithmic space.

\paragraph{Plan of the paper}

We recall from~\cite{Benedetti} the definitions of $\EAlam$ and $\muEAlam$ in
Section~\ref{sec:definitions}, and then quickly prove
Theorem~\ref{thm:aesthetic} in Section~\ref{sec:proof-mueal}. The bulk of the
paper is Section~\ref{sec:regular}, dedicated to proving
Theorem~\ref{thm:elam-reg}. The conclusion (Section~\ref{sec:conclusion})
discusses the above-mentioned new perspectives on $\EAlam$ opened up by our
results and by refinements of Lemma~\ref{lem:finsem}.

\paragraph{Acknowledgments}

This work owes a great deal to Thomas Seiller's supervision. Thanks also to
Patrick Baillot, Alexis Ghyselen, Damiano Mazza (an extremely fruitful
discussion with Thomas and him triggered this work) and Pierre Pradic.

\newpage

\section{The elementary affine $\lambda$-calculus}
\label{sec:definitions}

The syntax of elementary affine $\lambda$-terms and the reduction rules are
given by
\[ t, u ::= x \mid \lambda x.\, t \mid \lambda\oc x.\, t \mid t\,u \mid \oc t
  \qquad
 (\lambda x.\, t)\, u \longrightarrow_\beta t\{x := u\} \quad
  (\lambda\oc x.\, t)\, (\oc u) \longrightarrow_\oc t\{x := u\}
\]
where $x$ is taken in a countable set of variables, and $t\{x:=u\}$ refers to
the substitution of all free occurrences of $x$ in $t$ by $u$. The reduction
rules $\longrightarrow_\beta$ and $\longrightarrow_\oc$ are actually the
contextual closure of the rules given above, for the obvious notion of context
(see~\cite{Benedetti} for details).

We shall also write $\mathtt{let}\; \oc x \leftarrow u \;\mathtt{in}\; t$ for
$(\lambda\oc x.\, t)\, u$ (this is just some \enquote{syntactic sugar}). The
notion of \emph{depth} of a subterm in a term, defined as the number of
\emph{exponential modalities} $\oc(-)$ (\enquote{exponentials} for short)
surrounding the subterm, will play an important role.

As an example, let us formally define the Church-encoded binary strings:
\[ \text{for $w = w_1\ldots w_n \in \{0,1\}^*$},\quad \overline{w} =
  \lambda\oc f_0.\, \lambda\oc f_1.\, \oc(\lambda x.\, f_{w_1}\, (\ldots
  (f_{w_n}\,x)\ldots))\]

The above is essentially Simpson's linear $\lambda$-calculus with
thunks~\cite{Simpson}. (Other examples of linear $\lambda$-calculi with explicit
exponentials are given in~\cite{BangCalculus}.) We shall now turn this untyped
calculus into $\EAlam$ by endowing it with its type system -- an adaptation of
Coppola \emph{et al.}'s Elementary Type Assignment System~\cite{Coppola}. The
grammar of types for $\EAlam$ is
\[ A ::= \alpha \mid S \qquad S ::= \sigma \multimap \tau \mid \forall
  \alpha.\,S \qquad \sigma, \tau ::= A \mid \oc \sigma \] The two first classes
of types are called respectively \emph{linear} and \emph{strictly linear}. (We
follow the terminology of~\cite{Benedetti}; \enquote{linear} does not mean
exponential-free, it merely means that the head connective is not an
exponential.) The reason for restricting quantification to strictly linear types
is a technical subtlety related to subject reduction (see~\cite[\S7.2]{Coppola}).

The typing judgements involve a context split into three parts: they are of the
form $\Gamma \mid \Delta \mid \Theta \vdash t : \sigma$. The idea is that the
partial assignements $\Gamma$, $\Delta$ and $\Theta$ of variables to types
correspond respectively to linear, non-linear and \enquote{temporary} variables;
accordingly, $\Gamma$ maps variables to linear types (denoted $A$ above),
$\Delta$ maps variables to types of the form $\oc\sigma$, while $\Theta$ maps
variables to arbitrary types. The domains of $\Gamma$, $\Delta$ and $\Theta$ are
required to be pairwise disjoint. The derivation rules for $\EAlam$ are:
\[\text{variable rules}\qquad \frac{}{\Gamma, x : A \mid \Delta \mid \Theta
    \vdash x : A} \qquad \frac{}{\Gamma \mid \Delta \mid \Theta, x : \sigma \vdash x
    : \sigma}\]
\[\text{abstraction rules}\qquad \frac{\Gamma, x : A \mid \Delta \mid \Theta
    \vdash t : \tau}{\Gamma \mid \Delta \mid \Theta \vdash \lambda x.\, t : A
    \multimap \tau} \qquad \frac{\Gamma \mid \Delta, x : \oc\sigma \mid \Theta
    \vdash t : \tau}{\Gamma \mid \Delta \mid \Theta \vdash \lambda\oc x.\, t :
    \oc\sigma \multimap \tau} \]
\[\text{application rule\footnotemark}\qquad \frac{\Gamma \mid \Delta \mid \Theta
    \vdash t : \sigma \multimap \tau \quad \Gamma' \mid \Delta \mid \Theta
    \vdash u : \sigma}{\Gamma \uplus \Gamma' \mid \Delta \mid \Theta \vdash t\,u
    : \tau} \]
\footnotetext{$\Gamma \uplus \Gamma'$ means $\Gamma \cup \Gamma'$ with the
  assumption that the domains of $\Gamma$ and $\Gamma'$ are disjoint.}
\[\text{quantifier rules\footnotemark}\qquad \frac{\Gamma \mid \Delta \mid \Theta
    \vdash t : S}{\Gamma \mid \Delta \mid \Theta \vdash t : \forall \alpha.\,S}
  \qquad \frac{\Gamma \mid \Delta \mid \Theta \vdash t : \forall \alpha.\,
    S}{\Gamma \mid \Delta \mid \Theta \vdash t : S\{\alpha := A\}} \]
\footnotetext{In the introduction rule (left), $\alpha$ must not appear as a
  free variable in $\Gamma$, $\Delta$ and $\Theta$.}
\[\text{functorial promotion rule}\qquad \frac{\varnothing \mid \varnothing \mid
    \Theta \vdash t : \sigma}{\Gamma \mid \oc\Theta, \Delta \mid \Theta' \vdash
    \oc t : \oc\sigma} \]

In these rules, following the conventions established above, $A$ stands for a
linear type, $S$ stands for a strictly linear type and $\sigma$ and $\tau$ stand
for arbitrary types. In particular, in the quantifier elimination rule, $\alpha$
can only be instantiated by a linear type. So, for instance, one cannot give the
type $\oc\beta \multimap \oc\beta$ to $\lambda x.\, x$ through a quantifier
introduction followed by a quantifier elimination; indeed, as one would expect,
the only normal term of this type is $\lambda\oc x.\, \oc{x}$. (Despite this,
the polymorphism is still impredicative.)

Coming back to the example of Church binary strings, one can show by induction
that
\[ \text{for $w = w_1\ldots w_n \in \{0,1\}^*$},\quad x : \alpha \mid
  \varnothing \mid f_0 : \alpha \multimap \alpha,\, f_1 : \alpha \multimap
  \alpha \vdash f_{w_1}\, (\ldots (f_{w_n}\,x)\ldots) : \alpha \]
and deduce from this that $\vdash \overline{w} : \Str$ (recall that $\Str =
\forall \alpha.\, \oc(\alpha \multimap \alpha) \multimap \oc(\alpha \multimap
\alpha) \multimap \oc(\alpha \multimap \alpha)$).

The system $\muEAlam$ is obtained by extending the grammar of types with $S ::=
\ldots \mid \mu\alpha.\, S$, and adding new derivation rules for the type
fixpoint operator $\mu$:
\[ \text{$\mu$-fold}\quad\frac{\Gamma \mid \Delta \mid \Theta \vdash t :
    S\{\alpha := \mu\alpha.\,S\}}{\Gamma \mid \Delta \mid \Theta \vdash t :
    \mu\alpha.\,S} \qquad \text{$\mu$-unfold}\quad\frac{\Gamma \mid \Delta \mid
    \Theta \vdash t : \mu\alpha.\, S}{\Gamma \mid \Delta \mid \Theta \vdash
    t : S\{\alpha := \mu\alpha.\, S\}} \]

Let us recall two basic properties satisfied both by $\EAlam$ and $\muEAlam$,
all proved in~\cite{Benedetti}.

\begin{proposition}[Stratification and linearity {\cite[Lemma~27]{Benedetti}}]
  \label{prop:stratification}
  Let $t$ be a typable term.
  \begin{itemize}
  \item for any subterm of the form $\lambda\oc x.\, u$ of $t$, all the
    occurrences of $x$ must be at depth 1 in $u$;
  \item for any subterm $\lambda x.\, u$ of $t$, there is \emph{at most one}
    occurrence of $x$ in $u$, whose depth must be 0 in $u$.
  \end{itemize}
  As a consequence, the reduction rules are depth-preserving.
\end{proposition}

\begin{proposition}[$k$-fold functorial promotion
  {\cite[Proposition~28]{Benedetti}}]
  \label{prop:promotion}
  Let $t : \sigma_1 \multimap \ldots \multimap \sigma_n \multimap \tau$ is a
  closed elementary affine $\lambda$-term and $k \geq 1$. There is a term
  $t^{(k)} : \oc^{k}\sigma_1 \multimap \ldots \oc^{k}\sigma_n \multimap
  \oc^{k}\tau$ such that $t^{(k)}\,(\oc^{k}u_1)\,\ldots\,(\oc^{k}u_n)$ and
  $\oc^{k}(t\,u_1\,\ldots\,u_n)$ have the same normal form for all closed
  terms $u_i : \sigma_i$ ($i \in \{1, \ldots, n\}$).
\end{proposition}

\section{The $k$-FEXPTIME hierarchy in $\muEAlam$ (proof of
  Theorem~\ref{thm:aesthetic})}
\label{sec:proof-mueal}

First, the soundness part of Theorem~\ref{thm:aesthetic} follows immediately
from Theorem~\ref{thm:scottfp}.
\begin{proposition}
  All functions represented by $\muEAlam$-terms of type $\oc \Str \multimap
  \oc^{k+1} \Str$ are in \emph{$k$-FEXPTIME}.
\end{proposition}
\begin{proof}
  There exists a coercion $\oc\Str \multimap \oc^2\Str_S$ (by completeness part
  of Theorem~\ref{thm:scottfp} applied to the identity function in FP) which
  lifts by functorial promotion (Proposition~\ref{prop:promotion}) to
  $\oc^{k+1}\Str \multimap \oc^{k+2}\Str_S$. So any function represented by a
  term of type $\oc\Str \multimap \oc^{k+1}\Str$ is also represented by a term
  of type $\oc\Str \multimap \oc^{k+2} \Str_S$. Thus the soundness part of
  Theorem~\ref{thm:scottfp} applies.
\end{proof}

For the extensional completeness, we also take Theorem~\ref{thm:scottfp} as our
starting point. The idea is to convert $\oc\Str_S$ into $\Str$ with the help of
an auxiliary integer which provides an upper bound on the length of the string.
(Similar ideas appear in~\cite{Ghyselen}.)

We shall use the type of \emph{Church natural numbers} and the usual
second-order encoding of pairs:
\[\Nat = \forall \alpha.\, \oc(\alpha \multimap \alpha) \multimap \oc(\alpha
  \multimap \alpha) \qquad \sigma \otimes \tau = \forall
  \alpha.\, (\sigma \multimap \tau \multimap \alpha) \multimap \alpha\]
The aforementioned upper bound will be an inhabitant of the type $\Nat$. An
integer $n \in \mathbb{N}$ is represented in $\Nat$ by the iterator $f \mapsto
f^n$ (formally, $\overline{n} = \lambda\oc f.\, \oc(\lambda x.\, f\,(\ldots
(f\, x) \ldots))$ with $n$ times $f$).

To help readability we extend the syntax with the abbreviation
\begin{itemize}
\item $u \otimes v := \lambda f.\, f\,u\,v$ so that $u\otimes v :
  \sigma\otimes\tau$ if $u : \sigma$ and $v : \tau$
\end{itemize}
given in~\cite{Benedetti}, and introduce some additional syntactic sugar:
\begin{itemize}
\item $\mathtt{let}\; x \otimes y \leftarrow u\;\mathtt{in}\;t := u\,(\lambda
  x.\, \lambda y.\, t)$ for $u : \sigma \otimes \tau$, and $\lambda(x \otimes
  y).\,t := \lambda z.\,\mathtt{let}\; x \otimes y \leftarrow z\;\mathtt{in}\;t$
\item $\mathtt{case}\; u \mid \mathtt{0}x \mapsto a \mid \mathtt{1}y \mapsto b
  \mid \varepsilon \mapsto c := u\,(\lambda x.\, a)\,(\lambda y.\, b)\,c$ for $u
  : \Str_S$
\end{itemize}
The affine projections $\pi_i = \lambda(x_1 \otimes x_2).\, x_i$ ($i \in
\{1,2\}$) are also defined in~\cite{Benedetti}.
\begin{remark}
  Our definition of $\lambda(x \otimes y).\,t$ is much simpler that the one
  given in~\cite{Benedetti}, but the drawback is that it only works when the
  type of $t$ is linear, i.e.\ its head connective is not an exponential.
  Indeed, $u : \sigma \otimes \tau$ can be instantiated to $u : (\sigma
  \multimap \tau \multimap A) \multimap A$ by the quantifier elimination rule
  only when $A$ is linear. This condition will hold in our use cases below.
\end{remark}

Now that we are equipped with all these data types, we can make progress on our proof.

\begin{lemma}
  \label{lem:cast}
  There exists a $\muEAlam$-term $\mathtt{cast} : \Nat \multimap \oc \Str_S
  \multimap \Str$ which converts a Scott encoding into a Church encoding, provided
  that the integer argument is greater or equal to the length of the string.
\end{lemma}
\begin{proof}

  Our implementation of $\mathtt{cast}$ instantiates the input $\Nat$ on
  $(\alpha \multimap \alpha) \otimes \Str_S$ where $\alpha$ is the eigenvariable
  of the $\forall$ in the output $\Str$ (recall that $S(\varepsilon)$ refers to the
  Scott encoding of the empty word):
  \[ \mathtt{cast} = \lambda n.\,\lambda\oc w.\,\lambda\oc f_0.\, \lambda\oc
    f_1.\, \mathtt{let}\; \oc g \leftarrow n\, \oc(\lambda(h \otimes u).\, t)
    \;\mathtt{in}\; \oc(\pi_1\,(g\, ((\lambda x.\, x) \otimes w))) \]
  \[   \text{with}\ t = \mathtt{let}\; f \otimes v \leftarrow
  (\mathtt{case}\; u \mid \mathtt{0}v \mapsto f_0 \otimes v \mid \mathtt{1}v \mapsto
    f_1 \otimes v \mid \varepsilon \mapsto (\lambda z.\, z) \otimes S(\varepsilon))
  \;\mathtt{in}\; (\lambda x.\,
  h\,(f\,x)) \otimes v
  \]
  To explain this functional program, let us reformulate it as an imperative
  algorithm: $t$ can be considered as the body of a \texttt{for} loop which
  alters two mutable variables $h : (\alpha \multimap \alpha)$ and $u : \Str_S$.
  At each iteration, if $u$ is non-empty, its first letter is popped (viewing
  $u$ as a mutable stack) and $h$ is post-composed with either $f_0$ or $f_1$
  depending on this letter.
  
  After $n$ iterations starting from $h = (\lambda x.\, x)$ and $u = w$, if $w$
  is the Scott encoding of $w_1 \ldots w_m$, the result obtained is $(f_{w_1}
  \circ \ldots \circ f_{w_N}) \otimes (S(w_{N+1} \ldots w_m))$ where $N =
  \min(n,m)$. In particular, if $n \geq m$, the first component will be $f_{w_1}
  \circ \ldots \circ f_{w_m}$ -- which corresponds to the definition of the
  Church encoding.
\end{proof}

To obtain the desired upper bound, we recall a lemma from~\cite{Benedetti}. It
is used in the proof of Theorem~\ref{thm:scottfp} in order to simulate Turing
machines.

\begin{lemma}[{\cite{Benedetti}}]
  Let $\mathcal{M}$ be a $k$-FEXPTIME Turing machine. There is a $\EAlam$-term
  $t_{\mathcal{M}} : \oc\Str \multimap \oc^{k+1}\Nat$ computing an upper bound
  on the running time of $\mathcal{M}$ on the given input string.
\end{lemma}

We now have all the ingredients for the extensional completeness proof.

\begin{theorem}
  All \emph{$k$-FEXPTIME} functions can be represented by $\muEAlam$-terms of
  type $\oc\Str \multimap \oc^{k+1}\Str$.
\end{theorem}
\begin{proof}
  Consider any function computed by a $k$-FEXPTIME Turing machine $\mathcal{M}$.
  By the completeness part of Theorem~\ref{thm:scottfp}, we can choose a
  $\muEAlam$-term $f : \oc \Str \multimap \oc^{k+2} \Str_S$ computing this
  function. We also choose a term $t_{\mathcal{M}}$ satisfying the conditions of
  the above lemma. Then the term
  \[ \lambda\oc w.\, \mathtt{cast}^{(k+1)}\,(t_{\mathcal{M}}\,\oc w)\,(f\,\oc w)
    : \oc\Str \multimap \oc^{k+1}\Str \]
  -- where $\mathtt{cast}^{(k+1)}$ is the $(k+1)$-fold functorial promotion of
  $\mathtt{cast}$ -- computes the same function as $\mathcal{M}$. Indeed, the
  assumption of Lemma~\ref{lem:cast} is satisfied, since for a Turing machine,
  the length of the output is bounded by the running time.
\end{proof}

\section{Regular languages in $\EAlam$ (proof of Theorem~\ref{thm:elam-reg})}
\label{sec:regular}

In this section, we wish to show that, in $\EAlam$ (\emph{without} fixpoints):
\begin{itemize}
\item all terms $t : \oc\Str \multimap \oc\oc\Bool$ decide regular languages;
\item moreover, all regular languages can be decided by terms $t : \Str
  \multimap \oc\Bool$.
\end{itemize}
By functorial promotion, the class of languages characterized by $\Str \multimap
\oc\Bool$ is included in the class corresponding to $\oc\Str \multimap
\oc\oc\Bool$, so this will entail that both are exactly the class of regular
languages. The situation is the opposite of the previous section: the second
item (extensional completeness) is easy, while the first (soundness) is hard.

Regular languages admit many well-known equivalent definitions, e.g.\ regular expressions
and finite automata (with many variants: non-determinism, bidirectionality,
etc.). The classic characterization which will prove useful for us is:
\begin{theorem}
  \label{thm:monoid}
  A language is regular if and only if it can be expressed as $\varphi^{-1}(S)$,
  where $\varphi : \{0,1\}^* \to M$ is a monoid morphism, $M$ is a \emph{finite}
  monoid and $S \subseteq M$.
\end{theorem}

\subsection{Extensional completeness}

\begin{proposition}
  All regular languages can be decided by $\EAlam$-terms of type $\Str \multimap
  \oc\Bool$.
\end{proposition}
\begin{proof}
  Let $\varphi : \{0,1\}^* \to M$ be a morphism to a finite monoid $M$. Without
  loss of generality, we may assume that the underlying set of $M$ is $\{1,
  \ldots, k\}$, and the identity element of the monoid is $1$. We represent the
  monoid elements in $\EAlam$ as inhabitants of the type $\mathtt{M} = \forall
  \alpha.\, \alpha \multimap \ldots \alpha \multimap \alpha$; the element $i$ is
  mapped to the term $m_i = \lambda x_1.\, \ldots\, \lambda x_k.\, x_i$. We
  define:
  \begin{itemize}
  \item $\delta_c = \lambda m.\,m\,m_{\varphi(c)\cdot 1}\,
    \ldots\,m_{\varphi(c) \cdot k} : \mathtt{M} \multimap \mathtt{M}$ for $c
    \in \{0,1\}$
  \item for $S \subseteq M$, $\chi_S = \lambda m.\, m\,b_1\,\ldots\,b_k :
    \mathtt{M} \multimap \Bool$ where $b_i = \mathtt{true}$ (resp.
    $\mathtt{false}$) if $i \in S$ (resp.\ $i \notin S$).
  \end{itemize}
  Then the language $\varphi^{-1}(S)$ is decided by the term $\lambda w.\,
  \mathtt{let}\; \oc d \leftarrow w\,\oc\delta_0\,\oc\delta_1 \;\mathtt{in}\;
  \oc(\chi_S\,(d\,m_1))$.
\end{proof}

Next, to prepare the ground for our proof of soundness in $\EAlam$, we review our
direct inspiration in the simply typed $\lambda$-calculus: the proof of one
direction of Theorem~\ref{thm:hk}. The goal is to show that any simply typed
$\lambda$-term $t : \Str_{\STlam}[A] \to \Bool_\STlam$, where $A$ is an
arbitrary simple type, decides a language $\Lang_\STlam(t)$ which is
\emph{regular}. This was done using automata in~\cite{HillebrandKanellakis}, but
we find it simpler to work with monoid morphisms (though this is, in the end,
merely a different presentation of the same proof).

\subsection{A short soundness proof for Hillebrand and Kanellakis's theorem
  (sketch)}

We shall omit the subscripts in the types $\Str_\STlam[A]$ and $\Bool_\STlam$ in
this subsection.

Let us fix a simple type $A$. The fundamental idea is that, given any
\emph{denotational semantics} $\denot{-}$:
\begin{itemize}
\item the denotation $\denot{\overline{w}} \in \denot{\Str[A]}$ of the encoding
  of $w \in \{0,1\}^*$ is enough to determine $\denot{t\,\overline{w}} \in
  \denot{\Bool}$ -- this is simply the compositionality of the semantics;
\item provided the semantics is \emph{non-trivial}, i.e.\ $\denot{\mathtt{true}}
  \neq \denot{\mathtt{false}}$, this subsequently determines $t\, \overline{w}$.
\end{itemize}
Formally, let us define $\varphi_A : \{0,1\}^* \to \denot{\Str[A]}$ by $\varphi_A(w) =
\denot{\overline{w}}$; then if $\denot{-}$ is non-trivial,
\[ \Lang_\STlam(t) = \varphi_A^{-1}(\{\omega \in \denot{\Str[A]} \mid
  \denot{t}(\omega) = \denot{\mathtt{true}} \}) \]
To show that $\Lang_\STlam(t)$ is regular, we shall apply
Theorem~\ref{thm:monoid} to this equation. We must make sure that:
\begin{itemize}
\item $\denot{\Str[A]}$ can be endowed with a monoid structure, in such a way
  that $\varphi$ is a monoid morphism -- this is caused by the use of
  \emph{Church encodings};
\item $\denot{\Str[A]}$ is finite -- thanks to the existence of a \emph{finite
    semantics} for the simply typed $\lambda$-calculus.
\end{itemize}

Our choice of semantics, to satisfy both conditions, is the usual interpretation
of types by mere sets (called the \enquote{full type frame}
in~\cite{HillebrandKanellakis}): $\denot{A \to B} = \denot{B}^{\denot{A}}$, with
$\denot{o} = \{0,1\}$ for the base type. Any choice for $\denot{o}$ with at
least two elements makes the semantics non-trivial. Furthermore, since
$\denot{o}$ is finite, the denotations of all types are also finite.

Finally, in order to define a monoid structure on $\denot{\Str[A]}$, observe that
\[ \denot{\Str[A]} = \left(\denot{A \to A}^{\denot{A \to A}}\right)^{\denot{A
      \to A}} \cong \End(\denot{A})^{\End(\denot{A})^2} \]
where $\End(\denot{A})$ is the monoid of maps from $\denot{A}$ to itself, endowed
with function composition. Thus, the right-hand side can be seen as a product of
monoids. Proving that $\varphi$ is a morphism can then be done componentwise;
the condition to be checked can be expressed as:
\[ \forall (f_0, f_1) \in \End(\denot{A})^2,\; \left(w \mapsto
    \denot{\overline{w}}(f_0,f_1)\right) \text{ is a morphism } \{0,1\}^*
  \to \End(\denot{A}) \]
By definition, $\overline{w} = \lambda f_0.\, \lambda f_1.\, \lambda x.\,
f_{w_1}\,(\ldots (f_{w_n}\, x) \ldots)$ (where $w = w_1 \ldots w_n$) so
\[ \forall (f_0, f_1) \in \End(\denot{A})^2,\; \denot{\overline{w}}(f_0,f_1) =
  f_{w_1} \circ \ldots \circ f_{w_n} \]
therefore $\varphi$ is none other than the product, over all $(f_0,f_1)
\in \End(\denot{A})^2$, of the monoid morphisms $\{0,1\}^* \to \End(\denot{A})$
defined by $c \mapsto f_c$ for $c \in \{0,1\}$.

\begin{remark}
  This reasoning can be made to work with any finite semantics of $\STlam$, not
  just sets. An interesting choice is the \enquote{linearized Scott
    model}\footnote{This model is obtained from a semantics of linear logic as
    its exponential co-Kleisli category, i.e.\ via the translation $A \to B :=
    \oc{A} \multimap B$. The resulting category embeds fully and faithfully into
    the usual category of Scott domains and continuous functions, hence the
    name. See~\cite{TeruiSemantics} for a short self-contained definition.}: as
  remarked by Terui~\cite{TeruiSemantics}, in that semantics, the points in the
  denotation of a Church-encoded word correspond to nondeterministic finite
  automata accepting that word. This idea is also at the heart of Grellois and
  Melliès's semantic approach to higher-order model
  checking~\cite{GrelloisMellies,grellois}.
\end{remark}

\subsection{Soundness for regular languages in $\EAlam$}

Our goal is now to emulate the above proof to show that the $\EAlam$-terms of
type $\oc\Str\multimap\oc\oc\Bool$ decide regular languages. (The result for
$\Str\multimap\oc\Bool$ then follows by functorial promotion.) While the core of
the semantic evaluation argument is similar, we need to do some syntactic
analysis first before coming to this point.

\subsubsection{Some lemmas and a truncation operation}

Our proof relies on some general properties of $\EAlam$. The two following ones
were established in~\cite{Benedetti}.

\begin{proposition}[Reading property for booleans
  {\cite[Lemma~31(i)]{Benedetti}}]
  \label{prop:reading}
  The only closed inhabitants of the type $\oc\oc\Bool$ are
  $\oc\oc\mathtt{true}$ and $\oc\oc\mathtt{false}$. ($\mathtt{true} = \lambda
  x.\, \lambda y.\, x$ and $\mathtt{false} = \lambda x.\, \lambda y.\, y$)
\end{proposition}

\begin{proposition}[$\oc$-inversion {\cite[Lemma~29(i)]{Benedetti}}]
  \label{prop:bang-inv}
  If $\varnothing \mid \Delta \mid \varnothing \vdash t : \oc\sigma$, then $t =
  \oc{t'}$ for some term $t'$.
\end{proposition}

We will also make use of a \emph{truncation} operation on $\EAlam$-terms. (To
our knowledge, it has not appeared previously in the literature.) Its purpose is
to erase all exponentials. This will be how the \emph{stratification} property
of $\EAlam$ (cf.\ Proposition~\ref{prop:stratification}) comes into play.

\begin{definition}
  The \emph{truncation at depth 0} $\trunc{-}$ is defined inductively on terms
  as:
  \[ \trunc{\oc{t}} = (\lambda x.\, x) \quad \trunc{(\lambda\oc x.\,t)} =
    \lambda x.\, \trunc{t} \quad \trunc{\lambda x.\, t} = \lambda x.\, \trunc{t}
    \quad \trunc{t\,u} = \trunc{t}\,\trunc{u} \quad\trunc{x} = x
  \]
  and on types as (using the abbreviation\footnote{This is justified as $\forall
    \alpha.\, \alpha \multimap \alpha$ is the unit to the tensor product used in
    Section~\ref{sec:proof-mueal}.} $1 = \forall \alpha.\, \alpha \multimap
  \alpha$):
  \[ \trunc{\oc\sigma} = 1 \quad \trunc{\sigma \multimap \tau} =
    \trunc{\sigma} \multimap \trunc{\tau} \quad \trunc{\alpha} = \alpha
    \quad \trunc{\forall \alpha.\, \sigma} = \forall \alpha.\,
    \trunc{\sigma} \]
\end{definition}

\begin{proposition}
  If a typing judgment $\Gamma \mid \Delta \mid \varnothing \vdash t : \sigma$
  is derivable in $\EAlam$, then, writing $\trunc{\Gamma}$ for $x_1 :
  \trunc{\tau_1}, \ldots, x_n : \trunc{\tau_n}$ if $\Gamma = x_1 : \tau_1,
  \ldots, x_n : \tau_n$, the judgment $\trunc{\Gamma} \mid \varnothing \mid
  \varnothing \vdash \trunc{t} : \trunc{\sigma}$ is derivable. In particular, if
  $t : \sigma$ is a closed term, then $\trunc{t} : \trunc{\sigma}$.
\end{proposition}
\begin{proof}
  By a mostly straightforward induction on the type derivation. Even so, let us
  treat a case involving a small subtlety: when the derivation ends with a
  quantifier elimination. In that case, the induction hypothesis gives us the
  typing judgment $\trunc{\Gamma} \mid \varnothing \mid \varnothing \vdash
  \trunc{t} : \forall \alpha.\, \trunc{S}$, and from this we must derive
  $\trunc{\Gamma} \mid \varnothing \mid \varnothing \vdash t : \trunc{S\{\alpha
    := \sigma\}}$. What the same instantiation rule can give us from our premise
  is $\trunc{\Gamma} \mid \varnothing \mid \varnothing \vdash t :
  \trunc{S}\{\alpha := \trunc{\sigma}\}$. One is therefore led to hope that
  $\trunc{S}\{\alpha := \trunc{\sigma}\} = \trunc{S\{\alpha := \sigma\}}$.
  Indeed, this can be checked by distinguishing, for each occurrence of $\alpha$
  in $\sigma$, two possible cases: either it is at depth 0 and remains in
  $\trunc{S}$, or at depth $\geq 1$ and is erased in $\trunc{S}$.
\end{proof}
\begin{remark}
  The above proof is the reason why we do not generalize here our truncation
  operation to a \enquote{truncation at depth $k$} for $k \geq 1$, which would
  erase all exponentials of depth $> k$. Indeed, a typical example for which the
  above reasoning would fail is the truncation at depth 1 of $\forall \alpha.\,
  \oc\alpha \multimap \alpha$ instantiated with $\alpha:=\oc\tau$. So these
  higher depths truncations would need additional conditions to be well-behaved.
\end{remark}

\begin{proposition}
  For all $k \in \mathbb{N}$ and all $\EAlam$-terms $t,t'$, if $t
  \longrightarrow t'$, then $\trunc{t} \longrightarrow^* \trunc{t'}$ (this
  also applies to untyped terms which satisfy the stratification property, i.e.\
  the conclusions of Proposition~\ref{prop:stratification}).
\end{proposition}
\begin{proof}
  If the redex contracted in $t$ to obtain $t'$ is at depth $\geq 1$, then one can
  prove that $\trunc{t} = \trunc{t'}$.

  Otherwise, by induction on the context of the redex, one can restrict to the
  case where $t = u\,v$ and the application of $u$ to $v$ is the contracted
  redex. We proceed by case analysis:
  \begin{itemize}
  \item If $u = \lambda x.\, u'$, then $t' = u'\{x:=v\}$. We use the fact that $x$
    appears only at depth zero in $u'$ (Proposition~\ref{prop:stratification}) to
    show that $\trunc{u'\{x:=v\}} = \trunc{u'}\{x:=\trunc{v}\}$. The
    latter is a reduct of $\trunc{u}\,\trunc{v} = \trunc{t}$.
  \item If $u = \lambda\oc x.\, u'$, then $v = \oc{v'}$ and $t' = u'\{x:=v'\}$.
    Moreover, $x$ appears only at depth 1 in $u'$ (again by
    Proposition~\ref{prop:stratification}). Therefore, $\trunc{u'}$ does not
    contain $x$ as a free variable; thus, $\trunc{t'} = \trunc{u'}$ is a reduct
    of $\trunc{t} = (\lambda x.\, \trunc{u'})\,\trunc{v}$.\qedhere
  \end{itemize}
\end{proof}

A final general observation (unrelated to truncation) before delving into the
soundness proof itself:
\begin{proposition}
  \label{prop:linearize-var}
  Let $t$ be a term a free variable $x$. Suppose that $t =
  t'\{x_1:=x\}\ldots\{x_n:=x\}$, where each $x_i$ appears only once in $t'$ (so
  $n$ is the number of occurrences of $x$ in $t$), and $\Gamma \mid \Delta \mid
  \Theta, x: A \vdash t : \tau$, where $A$ is linear (i.e. not of the form
  $\oc\sigma$). Then $\Gamma, x_1 : A, \ldots, x_n : A \mid \Delta \mid \Theta
  \vdash t : \tau$. We write $t' = t\{x:=x_1,\ldots,x_n\}$ for this situation.
  (Such a $t'$ always exists given $t$.)
\end{proposition}
\begin{proof}
  By induction on typing derivations, replacing each rule of the form $\ldots
  \mid \ldots \mid \ldots, x : A \vdash x : \sigma$ by a rule of the form
  $\ldots, x_i : A, \ldots \mid \ldots \mid \ldots \vdash x_i : A$ for some $i
  \in \{1, \ldots, n\}$.
\end{proof}

\subsubsection{Syntactic analysis}

We can now start looking at the languages decided by $\EAlam$-terms.

\begin{lemma}
  \label{lem:shape}
  For any $\EAlam$-term $t : \oc\Str\multimap\oc\oc\Bool$, there exists $u :
  \Str[\sigma_1] \multimap \ldots \multimap \Str[\sigma_n] \multimap \oc\Bool$
  (for some $n \in \mathbb{N}$) such that, for all $s : \Str$, $t\,\oc s$ and
  $\oc(u\,s \ldots\, s)$ have the same normal form.
\end{lemma}
\begin{proof}
  First, one may take $t$ to be in normal form. In that case, the only possible
  redex in $t\,\oc{s}$ is the application at the root. Moreover, $t\,\oc{s}$
  must be reducible since it is neither $\oc\oc\mathtt{true}$ nor
  $\oc\oc\mathtt{false}$, cf.\ Proposition~\ref{prop:reading}. Therefore, $t =
  (\lambda\oc x.\, t')$ (the case $t = (\lambda x.\, t')$ can be excluded
  because then $t$ would be of type $A \multimap \tau$ where $A$ is linear, in
  particular $A \neq \oc\Str$).

  The next step is to prove that $t = \lambda\oc x.\, \oc{t''}$ for some
  $\EAlam$-term $t''$. According to the typing rules, the judgment $\varnothing
  \mid \varnothing \mid \varnothing \vdash t : \oc\Str \multimap \oc\oc\Bool$
  can only be proven by first establishing $\varnothing \mid x : \oc\Str \mid
  \varnothing \vdash t' : \oc\oc\Bool$. According to the $\oc$-inversion
  property (Proposition~\ref{prop:bang-inv}), since the first and third part of
  the typing context are empty and the head connective of the type is `$\oc$',
  $t'$ must be of the form $\oc{t''}$.

  Finally, since $\varnothing \mid \varnothing \mid x : \Str \vdash t' :
  \oc\oc\Bool$ must hold (it is the only premise which can lead to the above
  judgement on $t'$), we can apply Proposition~\ref{prop:linearize-var} to $t''$
  (indeed, the type $\Str$ is linear). Then, the term $u = \lambda x_1.\,
  \ldots\, \lambda x_m.\, t''\{x:=x_1,\ldots,x_m\}$ (where $x$ occurs $m$ times
  in $t''$) enjoys the property claimed in the lemma statement.
\end{proof}

Let us focus on the case $n=1$ for a moment, and do the same kind of analysis
again.

\begin{lemma}
  \label{lem:unary}
  Let $u : \Str[\sigma] \multimap \oc\Bool$ be an $\EAlam$-term, and let $\tau =
  \sigma \multimap \sigma$.

  There exist $\EAlam$-terms $f_0 : \tau$, $f_1 : \tau$ and $g : \tau \multimap
  \ldots \multimap \tau \multimap \oc\Bool$ (with $m$ times $\tau$, for some $m
  \in \mathbb{N}$) such that for all $s : \Str$, if $s\,\oc{f_0}\,\oc{f_1}
  \longrightarrow^* \oc{h}$, then $u\,s$ and $\oc(g\,h\,\ldots\,h)$ have the
  same normal form.
\end{lemma}

\begin{proof}
  We assume that $u$ is in normal form. Since the head connective of
  $\Str[\sigma]$ is not `$\oc$', $u = \lambda x.\, v$ and $x : \Str[\sigma] \mid
  \varnothing \mid \varnothing \vdash v : \oc\Bool$. We may assume that $v$
  contains $x$ as a free variable; otherwise, $u$ is a constant function and the
  conclusion we want holds trivially (take $m = 0$).

  Let us examine in general the shape of $v : \oc\theta$ in normal form (where
  $\theta$ is not necessarily $\Bool$) such that $x$ appears free in $v$ and $x :
  \Str[\sigma] \mid \varnothing \mid \varnothing \vdash v : \oc\theta$.
  By~\cite[Lemma~29(ii)]{Benedetti}, $v$ must be an application: $v =
  p\,q_1\,\ldots\,q_k$ where $p$ is not an application and $k \geq 1$. Observe
  that $p$ cannot be of the form $\lambda y.\, p'$, since $p\,q_1$ would then be
  a redex. There are two possible cases:
  \begin{itemize}
  \item $p = x$, and then $\theta = \sigma \multimap \sigma = \tau$, $k=2$ and
    the closed $\EAlam$-terms $q_1,q_2 : \oc\tau$ must be of the form $q_i =
    \oc{q'_i}$ by $\oc$-inversion (Proposition~\ref{prop:bang-inv})
  \item $p = (\lambda\oc y.\, p')$, in which case $x$ must appear free in
    $q_1$. Indeed, suppose for the sake of contradiction that $q_1$ is closed;
    then $\varnothing \mid \varnothing \mid \varnothing \vdash q_1 : \theta_1 =
    \oc\rho$ for some $\rho$, therefore $\oc$-inversion gives us $q_1 = \oc{r}$
    for some $r$, so $p\,q_1$ would be a redex.
  \end{itemize}
  In the second case, we may furthermore take $k = 1$ w.l.o.g.: if $k \geq 2$,
  then for all $s : \Str[\sigma]$, the term $((\lambda\oc y.\,
  p'\,q_2\,\ldots\,q_k)\,q_1)\{x:=s\}$ has the same normal form as $v\{x:=s\}$
  (this is analogous to Regnier's \mbox{$\sigma$-equivalence} by redex
  permutations~\cite{Regnier}). And since $\varnothing \mid y : \oc\rho \mid
  \varnothing \vdash p'\,q_2\,\ldots\,q_k : \oc\theta$, the normal form of
  $p'\,q_2\,\ldots\,q_k$ is of the form $\oc{p''}$ (this is again an application
  of $\oc$-inversion).

  To recapitulate, either $v = x\,\oc{f_0}\,\oc{f_1}$ or $v = (\lambda\oc y.\,
  \oc{p})\,v' = \mathtt{let}\; \oc y \leftarrow v'\;\mathtt{in}\; \oc{p}$ where
  $x$ appears free in $v'$, but not in $p$. In the latter case, we have $x :
  \Str[\sigma] \mid \varnothing \mid \varnothing \vdash v' : \oc\theta'$. So, by
  induction on the size of terms,
  \[ v = \mathtt{let}\; \oc y_1 \leftarrow (\ldots (\mathtt{let}\; \oc y_l
    \leftarrow x\,\oc{f_0}\,\oc{f_1}\;\mathtt{in}\; \oc{p_l})
    \ldots)\;\mathtt{in}\; \oc{p_1} \]
  As a consequence, for all $s : \Str$, if $s\,\oc{f_0}\,\oc{f_1}
  \longrightarrow^* \oc{h}$ (recall that $f_0,f_1 : \tau$ are closed) then
  \[ u\,s = (\lambda x.\, v)\, s \longrightarrow v\{x:=s\} \longrightarrow^*
    \oc(p_1\{y_1:=(\ldots p_l\{y_l:=h\} \ldots)\}) \]

  (Morally, we are still trying to permute redexes; the reader may check that
  there is an analogy between the above operation and Carraro and Guerrieri's
  $V\,((\lambda x.\,L)\,N) \rightsquigarrow (\lambda x.\, V\,L)\,N$
  rule~\cite{CarraroGuerrieri} for the call-by-value $\lambda$-calculus.)

  Let $r = p_1\{y_1:=(\ldots p_l\{y_l:=z\} \ldots)\}$, where $z$ is a fresh
  variable, so that the right-hand side can be written as $\oc(r\{z:=h\})$.
  Since $h : \tau$ is a closed subterm of $\oc(r\{z:=h\}) : \oc\Bool$ (we are
  using subject reduction here), then it must be true that $\varnothing \mid
  \varnothing \mid z : \tau \vdash r : \Bool$. Let us now apply
  Proposition~\ref{prop:linearize-var}, using the fact that $\tau = \sigma
  \multimap \sigma$ is linear: for some $m \in \mathbb{N}$, $z_1 : \tau, \ldots,
  z_m : \tau \mid \varnothing \mid \varnothing \vdash r\{z:=z_1,\ldots,z_m\} :
  \Bool$.

  Finally, we take $g = \lambda z_1.\,\ldots\,\lambda z_m.\,
  r\{z:=z_1,\ldots,z_m\}$. The lemma statement holds with the $f_0$, $f_1$, $m$
  and $g$ that we have constructed.
\end{proof}

The last purely syntactic step is to use the truncation operation to formulate a
variation of the above lemma. The point is to be able to decide the membership
in the language defined by an $\EAlam$-term by computing purely in $\Alamtwo$.
This sets the stage for the use of a semantics of $\Alamtwo$.

\begin{lemma}
  \label{lem:unary-trunc}
  Let $u : \Str[\sigma] \multimap \oc\Bool$ be an $\EAlam$-term, and let $\tau =
  \trunc{\sigma \multimap \sigma}$.

  There exist $\Alamtwo$-terms $f_0 : \tau$, $f_1 : \tau$ and $g : \tau
  \multimap \ldots \multimap \tau \multimap \oc\Bool$ (with $m$ times $\tau$,
  for some $m \in \mathbb{N}$) such that for all $w \in \{0,1\}^*$, if
  $\overline{w}\,\oc{f_0}\,\oc{f_1} \longrightarrow^* \oc{h}$, then
  $u\,\overline{w}$ and $\oc(g\,h\,\ldots\,h)$ have the same normal form.

  (Recall that $\overline{w} : \Str$ is the Church encoding of $w$ in $\EAlam$.)
\end{lemma}
\begin{proof}
  Thanks to the previous lemma, there exist $f'_0 : \sigma \multimap \sigma$,
  $f'_1 : \sigma \multimap \sigma$ and $g' : (\sigma \multimap \sigma) \multimap
  \ldots \multimap (\sigma \multimap \sigma) \multimap \oc\Bool$ with $m$ times
  $\tau$, for some $m \in \mathbb{N}$, such that the conclusion holds by
  replacing $\tau$ by $\sigma \multimap \sigma$ and $f_0,f_1,g$ by
  $f'_0,f'_1,g'$. The only issue is that $f'_0,f'_1,g'$ might not be in
  $\Alamtwo$. The idea is therefore to take $f_0 = \trunc{f'_0}$, $f_1 =
  \trunc{f'_1}$ and $g = \trunc{g'}$, and to check that this works.

  Let $h' : \sigma \multimap \sigma$ be such that
  $\overline{w}\,\oc{f'_0}\,\oc{f'_1} \longrightarrow^* \oc{h'}$. Since
  $\overline{w} = \lambda\oc a_0.\, \lambda\oc a_1.\, \oc(\lambda x.\,
  a_{w_1}\, (\ldots (a_{w_n}\,x)\ldots))$,
  \[\lambda x.\, f'_{w_1}\, (\ldots (f'_{w_n}\,x)\ldots) \longrightarrow^* h'
    \quad \text{and by truncation} \quad
    \lambda x.\, f_{w_1}\, (\ldots (f_{w_n}\,x)\ldots) \longrightarrow^*
    \trunc{h'} \]

  So if $h$ is such that $\overline{w}\,\oc{f_0}\,\oc{f_1} \longrightarrow^*
  \oc{h}$, then by confluence~\cite[Lemma~8]{Benedetti}, $h$ and $\trunc{h'}$
  have the same normal form. Therefore, $g\,h\,\ldots\,h$ and
  $g\,\trunc{h'}\,\ldots\,\trunc{h'}$ have the same normal form. But the latter
  is none other than $\trunc{g'\,h'\,\ldots\,h'}$.

  To conclude, observe that:
  \begin{itemize}
  \item the normal form of $\oc(g'\,h'\,\ldots\,h')$ is the same as that of
    $u\,\overline{w}$ by the previous lemma;
  \item by Proposition~\ref{prop:reading}, the normal form of
    $g'\,h'\,\ldots\,h'$ is some $b \in \{\mathtt{true},\mathtt{false}\}$;
  \item since $\trunc{b} = b$, $\oc\trunc{g'\,h'\,\ldots\,h'}$ has the same
    normal form as $u\,\overline{w}$.
  \end{itemize}
  By the discussion above, this means that $\oc(g\,h\,\ldots\,h)$ and
  $u\,\overline{w}$ have the same normal form, as desired.
\end{proof}

\subsubsection{Semantic evaluation}

We are now ready to conclude our proof of soundness by adapting Hillebrand and
Kanellakis's argument. Let $\denot{-}$ be any non-trivial \emph{finite}
semantics of $\Alamtwo$ -- the notion of finiteness we need is that $\denot{A}$
has finitely semantic inhabitants for all $\Alamtwo$ types $A$. (Equivalently,
if our semantics is a category with a terminal object $1$, we require
$\mathrm{Hom}(1,\denot{A})$ to be finite for all $A$.) Recall that although such
a semantics is a central ingredient in our proof, we have simply assumed its
existence, which is proved elsewhere (see Lemma~\ref{lem:finsem} and the
subsequent discussion).

\begin{definition}
  Let $A$ be a $\Alamtwo$ type. We define $\bloup{A}{w}(\gamma_0,\gamma_1) =
  \gamma_{w_1} \circ \ldots \circ \gamma_{w_n}$ for $w \in \{0,1\}^*$ and
  $(\gamma_0,\gamma_1) \in \End(\denot{A})^2$. In other words, $\Phi_{A} :
  \{0,1\}^* \to \End(\denot{A})^{\End(\denot{A})^2}$ is the monoid morphism
  sending $c \in \{0,1\}$ to $(\gamma_0,\gamma_1) \mapsto \gamma_c$.
\end{definition}
Here $\End(\denot{A})$ refers to the monoid of endomorphisms of $\denot{A}$ in
the semantics.

\begin{proposition}
  Let $w \in \{0,1\}^*$ and $\bar{w} : \Str$ be its encoding. For any $\Alamtwo$
  type $\sigma$ and $\Alamtwo$-terms $f_0, f_1 : \sigma \multimap \sigma$,
  $\bar{w}\,!f_0\,!f_1$ normalizes into some $\oc{h}$ with $h : \sigma \multimap
  \sigma$, and $\bloup{A}{w}(\denot{f_0},\denot{f_1}) = \denot{g}$.
\end{proposition}
\begin{proof}
  As in the case of the simply typed $\lambda$-calculus, this is by definition
  of the Church encoding.
\end{proof}

\begin{lemma}
  Let $u : \Str[\sigma_1] \multimap \ldots \multimap \Str[\sigma_n] \multimap
  \oc\Bool$.

  For all $w_1,\ldots,w_n \in \{0,1\}^*$, the normal form of
  $u\,\overline{w_1}\,\ldots\,\overline{w_n}$ is completely determined by the
  functions
  $\bloup{\trunc{\sigma_1}}{w_1},\ldots,\bloup{\trunc{\sigma_n}}{w_n}$. As a
  consequence, the following language is regular:
  \[ \{w \in \{0,1\}^* \mid u\,\overline{w}\,\ldots\,\overline{w}
    \longrightarrow^* \oc\mathtt{true} \} \]
\end{lemma}
\begin{proof}
  We start with the case $n=1$, in which $u : \Str[\sigma] \multimap \oc\Bool$.
  Let $f_0 : \tau$, $f_1 : \tau$ and $g : \tau$ be given by
  Lemma~\ref{lem:unary-trunc}, where $\tau = \trunc{\sigma \multimap \sigma} =
  \trunc{\sigma} \multimap \trunc{\sigma}$. For all $w \in \{0,1\}^*$, if
  $\overline{w}\,f_0\,f_1 \longrightarrow^* \oc{h}$, then $u\,\overline{w}
  \longrightarrow^* b$ for some $b \in \{\mathtt{true},\mathtt{false}\}$ such
  that $g\,h\,\ldots\,h \longrightarrow^* b$. Since $f_0$, $f_1$ and $g$ are in
  $\Alamtwo$, so is $h$ (provided it is normal), and $\denot{g\,h\,\ldots\,h} =
  \denot{g}(\denot{h},\ldots,\denot{h})$ by compositionality. Therefore
  \[ \denot{b} = \denot{g}\left(\bloup{\trunc{\sigma}}{w}(\denot{f_0},\denot{f_1}),
    \ldots, \bloup{\trunc{\sigma}}{w}(\denot{f_0},\denot{f_1})\right) \]
  thanks to the previous proposition. Since our semantics is \emph{non-trivial},
  i.e.\ $\denot{b} = \denot{\mathtt{true}} \iff b = \mathtt{true}$,
  $\bloup{\trunc{\sigma}}{w}$ thus determines the normal form of
  $u\,\overline{w}$.

  The result for arbitrary $n \geq 1$ is obtained by induction on $n$ by
  repeatedly applying the case $n=1$.

  The consequence is that the language $\{w \in \{0,1\}^* \mid
  u\,\overline{w}\,\ldots\,\overline{w} \longrightarrow^* \oc\mathtt{true} \}$
  can be written using only conditions on $\bloup{\trunc{\sigma_i}}{w}$ ($i
  \in \{1,\ldots,n\}$), so it is the preimage of some subset of
  $\prod_{i=1}^n \End(\denot{\sigma_i})^{\End(\denot{\sigma_i})^2}$ by the
  monoid morphism $w \mapsto (\bloup{\trunc{\sigma_1}}{w}, \ldots,
  \bloup{\trunc{\sigma_n}}{w})$.
\end{proof}

This suffices to conclude the soundness proof. Let $t : \oc\Str \multimap
\oc\oc\Bool$. Then, by Lemma~\ref{lem:shape},
\[ \{ w \in \{0,1\}^* \mid t\, \oc\overline{w} \longrightarrow^*
  \oc\oc\mathtt{true} \} = \{w \in \{0,1\}^* \mid
  u\,\overline{w}\,\ldots\,\overline{w} \longrightarrow^* \oc\mathtt{true} \} \]
for some $u : \Str[\sigma_1] \multimap \ldots \multimap \Str[\sigma_n] \multimap
\oc\Bool$. The regularity of this language then follows from the above lemma.

\subsection{Overcoming the expressivity barrier}
\label{sec:overcoming}

Analyzing the our soundness proof for regular languages in $\EAlam$ reveals that
fundamentally, what restricts the computational power is a conjunction of two
facts:
\begin{enumerate}
\item the input $\Str$ is instantiated on some types $\sigma_1, \ldots, \sigma_n$ 
  \emph{known in advance};
\item these $\sigma_i$ are morally finite data types, since they admit finite
  semantics.
\end{enumerate}
This makes it impossible to iterate over, say, the configurations of a Turing
machine, since their size depends on the input and the type $\sigma_i$ cannot
\enquote{grow} to accomodate data of variable size.

If we stay at depth 2 in $\EAlam$, there is no way of avoiding the second fact
(one can always truncate the $\sigma_i$ to exponential-free types), so if we
want to retrieve a larger complexity class than regular languages without
resorting to type fixpoints, we should try to circumvent the first obstacle.
That means that the $\sigma_i$ should vary with the input. Thus, we are led to
consider that \emph{inputs should provide types}:
\begin{itemize}
\item the encoding of an input $x$ would be a term $t_x : \mathtt{Inp}[A_x]$,
  for some type $\mathtt{Inp}$ with one parameter;
\item this $t_x$ would then be given as argument to a program of type $\forall
  \alpha.\, \mathtt{Inp}[\alpha] \multimap \Bool$.
\end{itemize}
In other words, we are considering \emph{existential} input types. Indeed, if we
were to extend $\EAlam$ with existential quantifiers\footnote{The reason this
  extension is not incorporated is that existentials can be encoded:
  $\exists\alpha.\,\tau := \forall\beta.\, (\forall\alpha.\,\tau \multimap
  \beta) \multimap \beta$.}, there would be an isomorphism
$(\exists\alpha.\,\mathtt{Inp}[\alpha]) \multimap \Bool
\cong\forall\alpha.\,(\mathtt{Inp}[\alpha] \multimap \Bool)$.

\begin{remark}
  In fact there is a third fact which plays a role in bridling the complexity:
  the shape of the type $\Str$ which codes sequential iterations (but the same
  could be said of Church encodings of free algebras -- with such inputs one
  characterizes regular tree languages).

  For instance, let us consider as inputs circuits represented by the type
  \[ \forall X. \oc X \multimap \oc(X \multimap X \multimap X) \multimap \oc(X
    \multimap X \otimes X) \multimap \oc X \]
  where $\oc X$ corresponds to $\mathtt{true}$ constants, $\oc(X \multimap X
  \multimap X)$ corresponds to $\mathtt{nand}$ gates, and $\oc(X \multimap X
  \otimes X)$ corresponds to duplication gates used to represent fan-out. Then
  instantiating this with $X = \Bool$ and the obvious evaluation maps gives us
  an encoding of the circuit value problem, which is P-complete.

  Although this input type seems morally less legitimate than Church encodings,
  it is hard to pinpoint precisely why it should be rejected.
\end{remark}

\section{Conclusion}
\label{sec:conclusion}

This paper started with a positive result: there exists a characterization of FP
and $k$-FEXPTIME in $\muEAlam$ whose statement is very simple. However, the
characterization of regular languages in $\EAlam$, which takes up the rest of
the paper, could be seen as a negative result: it demonstrates the lack of
expressivity of $\EAlam$ without type fixpoints. (This is the spirit of
Section~\ref{sec:overcoming}.) Indeed, the small class of regular languages not
quite a well-behaved complexity class, e.g.\ it is not closed under
$\mathrm{AC}^0$ reductions.

That said, one can also read Theorem~\ref{thm:elam-reg} as positive evidence of
a connection between affine typing and automata. This connection clearly depends
on the use of Church encodings -- in other words, on the representation of
strings by their iterators. This opens up two avenues for investigation:
\begin{itemize}
\item One can search for other automata-theoretic classes of interest that can
  be characterized in $\EAlam$.
\item On the other hand, one can hope to obtain a well-behaved sub-polynomial
  complexity class by changing the representation of inputs, following the
  suggestions of Section~\ref{sec:overcoming}.
\end{itemize}

We are currently working on the first research direction, by attempting to
capture classes of \emph{transductions}, i.e.\ of functions computed by automata
with output. As of the time of writing, it seems likely that in $\EAlam$, $\Str
\multimap \Str$ captures the well-known class of \emph{regular functions}
(see~\cite{siglog} for an overview of classical transduction classes, including
regular functions), and that the class defined by $\oc\Str \multimap \oc\Str$
also admits an automata-theoretic characterization.

As for the second one, it is the topic of a sequel\footnote{This sequel has been
  published first, although the results in the present paper were mostly
  obtained before.} paper~\cite{sequel} (joint work with P.\ Pradic) which
studies an input type inspired by finite model theory, following Hillebrand's
thesis~\cite{HillebrandPhD}. We obtain what we believe to be a characterization
of \emph{deterministic logarithmic space} (L), and manage to prove that the
class we capture is between L and NL\footnote{Actually, a more precise upper
  bound is L with an oracle for \emph{unambiguous} non-deterministic logarithmic
  space.}.

\paragraph{The importance of semantics}

A novelty in our approach is that we betray the original spirit of
\enquote{light logics} such as Light Linear Logic and Elementary Linear
Logic~\cite{GirardELL}, which consisted in bounding the complexity of
normalization \enquote{geometrically}, independently of types. Here:
\begin{itemize}
\item geometry still plays an important structuring role, reflected by our use
  of a \enquote{truncation at depth zero} operation, which may be of independent
  interest;
\item but our fine-grained analysis also requires to take into account the
  influence of types through semantics.
\end{itemize}
Though we are not the first to apply semantics to obtain inexpressivity results
in light logics (cf.\ e.g.~\cite{DalLagoBaillot}), our recent discovery of a
finite semantics of linear polymorphism (cf.\ the discussion below the statement
of Lemma~\ref{lem:finsem}) opens up new possibilities. The above-mentioned
sequel on logarithmic space is an illustration of this new way of working in
$\EAlam$ and its variants: the best upper bound that we have is obtained using
the effectiveness of the second-order coherence space model studied
in~\cite{FiniteMALL2}.

\paragraph{Open questions}

Aside from the perspectives already mentioned, there is an obvious question that
remains after Theorem~\ref{thm:elam-reg}: what about $\oc\Str \multimap
\oc^{k+2}\Bool$ (resp.\ $\oc\Str \multimap \oc^{k+1}\Str$) for $k \geq 1$? For
now, what we know about the corresponding complexity class is that:
\begin{itemize}
\item it is contained in $k$-EXPTIME (resp.\ $k$-FEXPTIME), since the soundness
  results for $\muEAlam$ apply \emph{a fortiori} to $\EAlam$;
\item it contains $(k-1)$-EXPTIME (resp.\ $(k-1)$-FEXPTIME), by adapting the
  proofs given in~\cite{baillot}.
\end{itemize}
We must confess that we have no idea about what class $\oc\Str \multimap
\oc^{k+2}\Bool$ corresponds to, let alone about a proof strategy. Our only
guesses is that the first containment is strict, and that semantics can prove
useful for this problem.

\bibliographystyle{eptcs}
\bibliography{bi}

\end{document}